\documentclass[10pt]{article}
\usepackage[centering,margin=1in]{geometry}

\usepackage{times}

\usepackage{cite}
\usepackage{amsmath}
\usepackage{latexsym}
\usepackage{amsfonts}
\usepackage{amssymb}
\usepackage{flushend}

\usepackage{subcaption}

\usepackage{IEEEtrantools} %% To use IEEEqnarray Forms %%

\usepackage{amsthm}

\usepackage{nicefrac} % Nice inline Fractions
\usepackage{mathtools} % proper rendering of :=

\DeclarePairedDelimiter\abs{\lvert}{\rvert} %absolute value
\makeatletter
\let\oldabs\abs
\def\abs{\@ifstar{\oldabs}{\oldabs*}}

\newtheorem{theorem}{Theorem}

\theoremstyle{definition}
{\itshape}{\normalfont}

\theoremstyle{remark}

\theoremstyle{arequirement}

\theoremstyle{crequirement}
\newtheorem{crequirement}{Requirement}

\theoremstyle{metric}

\usepackage{xspace,xcolor}
\usepackage{paralist}
\usepackage{enumitem}
\usepackage{siunitx}
\DeclareSIUnit\knot{kn}

\usepackage{acronym}

\usepackage[pdftex]{graphicx}

\newcommand{\eg}{e.g.,\xspace}
\newcommand{\ie}{i.e.,\xspace}

\newcommand{\commentout}[1]{}

\usepackage{hyperref}
\hypersetup{
    colorlinks=true,	% false: boxed links; true: colored links
    linkcolor=blue,	% color of internal links
    citecolor=blue,	% color of links to bibliography
    filecolor=blue,    % color of file links
    urlcolor=blue		% color of external links
}

\def\BibTeX{{\rm B\kern-.05em{\sc i\kern-.025em b}\kern-.08em
    T\kern-.1667em\lower.7ex\hbox{E}\kern-.125emX}}

\newcommand{\domain}{X}
\newcommand{\codomain}{Y}

\newcommand{\inputvect}{\mathbf{x}}
\newcommand{\outputvect}{\mathbf{y}}

\newcommand{\futureinput}{\mathbf{j}}

\newcommand{\detvect}{\mathbf{d}}

\newcommand{\vectsize}{n}

\newcommand{\hitprob}{p_{\mathrm{hit}}}
\newcommand{\missprob}{p_{\mathrm{miss}}}

\newcommand{\prob}[1]{\mathrm{Pr}(#1)}

\newcommand{\minhit}{x_{\mathrm{min}}}
\newcommand{\minmiss}{y_{\mathrm{min}}}

\newcommand{\ProximityLoss}{\mathtt{LossProxAlrt}}
\newcommand{\ProximityMalfn}{\mathtt{ProxAlertMalfn}}

\newcommand{\failToDetect}{\mathtt{SgnDetFlr}}
\newcommand{\definedAs}{\stackrel{\text{def}}{=}}

\newcommand{\qso}{q_{\mathrm{tr}}}

\usepackage{fancyhdr}
\fancypagestyle{firstpage}
{
    \fancyhead[C]{\textcolor{red}{This pre-print is an expanded version of the peer-reviewed paper: \\G. Pai, ``Deriving Safety-related Performance Requirements for Machine Learnt Aeronautical Applications'', in Proceedings of the 44th AIAA DATC/IEEE Digital Avionics Systems Conference (DASC 2025), Montreal, Canada.}}   
    }

\begin{document}

\title{
Relating System Safety and Machine Learnt Model Performance\thanks{This work was authored by an employee of KBR, Inc., under Contract No.~80ARC020D0010 with the National Aeronautics and Space Administration. The United States Government retains and the publisher, by accepting the article for publication, acknowledges that the United States Government retains a non-exclusive, paid-up, irrevocable, worldwide license to reproduce, prepare derivative works, distribute copies to the public, and perform publicly and display publicly, or allow others to do so, for United States Government purposes. All other rights are reserved by the copyright owner.}}

\author{Ganesh J. Pai\\
	\normalsize{KBR / NASA Ames Research Center}\\
	\normalsize{Moffett Field, CA 94035, USA}\\
	\normalsize{ganesh.pai@nasa.gov}}

\date{}

\maketitle

\thispagestyle{firstpage}
\begin{abstract}

The prediction quality of machine learnt models and the functionality they ultimately enable (\eg object detection), is typically evaluated using a variety of quantitative metrics that are specified in the associated model performance requirements. When integrating such models into aeronautical applications, a top-down safety assessment process must influence both the model performance metrics selected, and their acceptable range of values. Often, however, the relationship of system safety objectives to model performance requirements and the associated metrics is unclear. 
Using an example of an aircraft emergency braking system containing a machine learnt component (MLC) responsible for object detection and alerting, this paper first describes a simple abstraction of the required MLC behavior. Then, based on that abstraction, an initial method is given to 
derive the minimum safety-related performance requirements, the associated metrics, and their targets for the both MLC and its underlying deep neural network, such that they meet the quantitative safety objectives obtained from the safety assessment process.
We give rationale as to why the proposed method should be considered valid, also clarifying the assumptions made, the constraints on applicability, and the implications for verification.

\end{abstract}

%\begin{IEEEkeywords}
%	Machine learning,
%	Model performance,
%	Performance requirements,  
%	Quantitative metrics, 
%	Safety performance, 
%	System safety objectives
%\end{IEEEkeywords}

\section{Introduction}\label{s:intro}

Amongst the core outcomes of the safety assessment process for civil aircraft~\cite{arp4761} are \emph{quantitative safety objectives} (QSOs). They represent an acceptable upper limit on the average probability of events that result in adverse safety effects. As part of aircraft system development, QSOs are allocated across the system hierarchy, from aircraft functions to the implementing \emph{items}. In conventional systems not including machine learning (ML), the decomposition, allocation, refinement, and eventual verification of QSOs has only been applied to hardware items. 

QSOs are not considered for software and the programmable aspects of hardware, in part, because the prevailing assurance guidelines~\cite{do178}, \cite{do254} intentionally avoid concepts of quantitative reliability or failure probability. Instead the focus is on applying process rigor to identify and correct development errors. The goal is providing assurance to an adequate level of confidence that the implementations of software or hardware designs are correct. 
The extent of the necessary development process rigor, given in terms of \emph{development assurance levels} (DALs), is proportional to the severity of the undesired effects identified from the safety assessment process. 

Although QSOs and DALs each address a different type of concern, they are associated through the severity of the effects of function (or item) failures. In particular, functions (or items) whose failures lead to higher severity effects are assigned a proportionally higher DAL and lower QSO, than those causing lower severity effects. Moreover, safety verification expects to confirm that those functions or items have been developed to the assigned level of rigor, and also that, for the related hardware, the QSOs have been attained. 

When integrating machine learning (ML)-based functionality into aircraft systems, it is anticipated that in addition to the assignment of DALs, allocating QSOs to \emph{machine learnt components} (MLCs), and relating the corresponding targets to the associated performance requirements and metrics will be mandated.\footnote{Although aviation industry consensus-based guidelines for development and assurance of aeronautical systems integrating ML are still being crafted, some regulatory publications expected to inform aviation rulemaking have proposed to relate QSOs to MLC and MLM performance metrics~\cite{easa-cp-L12}.}
Performance metrics for MLCs can be seen in part as quantitative criteria giving a long-term characterization (\ie over the duration of their intended use) of their behavior relative to their requirements. An MLC that fails to meet its requirements may lead to functional failures and thereby to system-level safety effects. As such, relating the performance metrics of an MLC to higher-level safety objectives facilitates capturing how it contributes to both safety and the overall functional intent.

Currently available guidance for integrating ML into aeronautical systems, \eg~\cite{easa-cp-L12}, does not clarify how valid MLC performance requirements should follow from an allocated QSO. Nor does it clarify how safety-related metrics may be selected, which metrics may be invoked in those requirements, or what range of values may be admissible for those metrics. Although those questions have been previously identified and investigated in other safety-critical domains (see Section~\ref{s:related-work} for related work), so far as we are aware they have not yet been adequately answered.
To that end, we adapt the \emph{aircraft emergency braking system} (AEBS) from prior literature~\cite{aebs-dasc2024} as an illustrative example (described in Section~\ref{s:example}), to make the following contributions in this paper:
\begin{compactitem}
\item We describe an initial method to translate QSOs obtained from a safety assessment into the safety-related performance requirements and associated metrics for an MLC (Section~\ref{s:methodology}). 

\item We develop an abstraction of the required behavior (Section~\ref{s:surrogate-model}) that: 
\begin{inparaenum}[(i)]
	\item traces to and meets the allocated QSO, and 
	\item is suitable for determining safety-related performance metrics and parameters such as the required confirmation threshold for detecting an object of interest in an image sequence, a tolerable miss ratio for not detecting that object, and the per image probability of non detection---a metric directly linked to the generalization capability of the machine learnt model.
\end{inparaenum}
\end{compactitem}

We supply the rationale to substantiate why our method and the resulting MLC performance requirements should be considered to be valid in Section~\ref{s:discussion}. 
This section also discusses the additional considerations that result from the method, the constraints that apply, and the implications for verification.

\section{Conceptual Background}\label{s:background} 

This section introduces the concepts relevant for the rest of the paper. 

\subsection{Quantitative Safety Objectives (QSOs)}

\newcommand{\pfh}{\mathrm{pfh}}

As mentioned earlier in Section~\ref{s:intro}, a QSO is an acceptable upper limit on the average per flight probability of an adverse event, usually normalized by exposure (itself expressed as a duration or a count). For systems and equipment installed on aircraft, QSOs may also be viewed as targets for reliability or, equivalently, probability of (the effects of) so-called \emph{failure conditions}: aircraft-level conditions that can directly or indirectly affect an aircraft and its occupants, including the crew, caused by one or more system failures, in combination with operational or environmental conditions encountered during various flight phases. 

Note that failure conditions are synonymous with \emph{hazards}, as used in other safety-critical domains. Also, failures include both \emph{loss of function} and \emph{malfunction}, whose respective causes encompass but are not limited to one or more component failures and their combinations, common causes, unintended or undesired emergent system interactions, and development errors (including errors in requirements) and their respective effects.

QSOs are defined and selected such that they are inversely proportional to the severity of the credible worst-case safety effects identified in a functional hazard assessment (FHA). The acceptable range of values for a QSO relative to effect severity is codified in civil aviation regulatory guidance documents. For example, a failure condition of $\mathsf{MINOR}$ severity, characterized as resulting in ``a slight reduction of functional capabilities or safety margins of the airplane, physical discomfort for passengers, or a slight increase in workload for the crew'', is associated with an allowable quantitative probability, \ie a QSO, of $10^{-3}$ per flight hour ($\pfh$)~\cite{ac23-1309}. 

The decomposition and allocation of a QSO in aircraft system development follows the preliminary system safety assessment (PSSA) process~\cite{arp4761}. That process contributes to a systematic evaluation of a system architecture to determine how failures of the architectural components lead to the failure conditions identified in the higher level FHA. The PSSA can employ different analysis techniques, such as fault tree analysis (FTA), and Markov models.

A quantitative, combinatorial FTA serves to validate the failure probability (or reliability) budgets established for architectural components, by confirming that they lead to a probability of an identified failure condition that is no worse than the associated QSO. Such a validation is a bottom-up assessment. A top-down analysis may also be performed to decompose and allocate the QSO to the architectural components by leveraging the fault tree logic and various heuristics. 

\commentout{
Although mentioned earlier (in Section~\ref{s:intro}), it is worth reiterating that such quantitative analyses are primarily undertaken only for the hardware elements of the system architecture. For software and the programmable aspects of hardware, the PSSA rather produces a \emph{qualitative} safety objective, representing the required level of development process rigor. 
%
%For MLCs, this dichotomy poses a conceptual challenge: an ML model is a mathematical formula constructed by applying learning algorithms to data~\cite{ml-nasa-tr-2024}. Typically, it may be specified as an executable software specification that is then optimized, and implemented using software items, programmable hardware items, or as a combination of the two. Effectively allocating a QSO to an MLC amounts to allocating it to software or programmable hardware items. 
%
As will be described (and justified) subsequently (Sections~\ref{s:surrogate-model} and~\ref{s:rationale}), the allocation of a QSO to the MLC hardware will be used as a reference to derive MLC and ML model performance requirements, even though the latter are to be evaluated for what is arguably an executable software specification. 
}

\subsection{Machine Learnt Models and Components and Their Characteristics}\label{ss:mlc-mlm-characteristics}

\subsubsection{Machine Learnt Model (MLM)}\label{sss:mlm}
An MLM is a mathematical formula or mapping rule, $f: \domain \to \codomain$, constructed by applying learning algorithms to (training) data, which comprises examples of the (patterns of) behavior to be learnt~\cite{ml-nasa-tr-2024}. Here, $\domain$ is the input space (or domain, or feature space), and $\codomain$ is the output space (or codomain, or space of responses). A deep neural network (DNN) is one possible such MLM. A description of $\domain$ as captured in the MLM requirements is known as an \emph{operational design domain} (ODD)~\cite{ml-nasa-tr-2024}. 

\subsubsection{Machine Learnt Component (MLC)}\label{sss:mlc}

In this paper, an MLC groups hardware and software implementations of one or more MLMs and, when appropriate, the supporting functionality (such as pre- and post-processing) necessary for their execution. An MLC is treated as a single entity allocated a DAL and a QSO from a system standpoint. 

\subsubsection{Deterministic Behavior}\label{sss:deterministic}

A trained MLM that does not continue to learn in use is \emph{static}. %has a fixed model structure and parameter values. 
That is, once $f$ has been constructed, it does not change given some future input $\futureinput \in \domain$. 
As such, $f$ is \emph{deterministic} in the sense that, given a specific input (vector) $\inputvect \in \domain$ for which the model produces a response (vector) $\outputvect \in \codomain$, any future input $\futureinput \in \domain$ that is identical to the input $\inputvect$ will always produce the same response $\outputvect$.

\subsubsection{Systematic Behavior and Correctness}\label{sss:systematic}

A suitable MLM is one that \emph{generalizes} from the training data inputs to unseen inputs from $\domain$, producing the required responses from $\codomain$.

The response $\outputvect$ for the input $\inputvect$ is \emph{correct} when it is the required response, otherwise it is \emph{incorrect}. More generally, because $f$ is deterministic, the responses of a static MLM to its inputs are \emph{systematic} in being correct or incorrect. That is, the input $\inputvect$ supplied at any future time point will always produce the same correct or incorrect response $\outputvect$. 
Moreover, if $g:\domain \to \codomain$ is the \emph{true} (but usually unknown) function relating the input and output spaces, then $f$ is correct when for all $x \in \domain, f(x) = g(x)$. That is, $f$ produces the correct response for any input from $\domain$, and is said to generalize perfectly. 

However, uncertainties in various aspects of the ML process, \eg epistemic uncertainty due to insufficient knowledge about the nature of $g$ and, therefore, a suitable form for $f$, as well as aleatoric uncertainty when sampling from $\domain$, together sampling limitations, can often result in an $f$ that may not always produce the correct responses for some subset of previously unseen inputs from $\domain$. 
Such imperfect generalization can be characterized in terms of the \emph{generalization error}, a (performance) metric of how MLM responses in use differ or deviate from the required responses for previously unseen inputs. 
The generalization error cannot be exactly calculated, but instead, theoretically, it can be probabilistically bounded to give a \emph{probably approximately correct} MLM~\cite{murphy-pml}, especially in the context of supervised learning (also see Section~\ref{ss:generalization}).

\subsubsection{Failure Probability and Insufficient Performance}
\label{sss:io-prob}

\newcommand{\ind}{\mathbf{1}}
\newcommand{\inputdist}{\mathrm{Pr}_{\domain}}
\newcommand{\inputspacedist}{\inputdist(\inputvect)}
\newcommand{\outputdist}{\mathrm{Pr}_{\codomain}}
\newcommand{\outputspacedist}{\outputdist(\outputvect)}
\newcommand{\conditionaldist}{\mathrm{Pr}_{\domain|\codomain}}

The inputs from $\domain$ may be governed in general by some (possibly unknown) generating process. The individual inputs can then be described in terms of (empirical estimates of) their limiting relative frequencies and, in turn, as a probability function $\inputspacedist$. In fact, a careful characterization of $\inputspacedist$ is a key requirement when defining the ODD~\cite{kape-safecomp-2023}. 
Given the preceding discussion (Sections~\ref{sss:deterministic} and~\ref{sss:systematic}), 
and assuming that $f$ is not a constant function (\ie $f$ produces the same response for any input), when the inputs occur according to $\inputspacedist$, the relative frequencies of the responses can also be established. In other words, the responses can be described through a probability function, $\outputspacedist$. 

Now, for a discrete input $\inputvect \in \domain$ occurring with a probability $\mathrm{Pr}_{\domain}(\inputvect)$, let $\outputvect \in \codomain$ be the correct response, and let $\ind_{f}(\inputvect)$ be an \emph{indicator function} defined such that $\ind_{f}(\inputvect) \equiv 1$ when $f$ returns an incorrect response (\ie $f(\inputvect) \neq \outputvect$), and is $0$ otherwise. Then, treating all incorrect responses as \emph{failures}, we can define a \emph{probability of failure} of an MLM as in \eqref{eq:prob-fail}, \ie the limiting relative frequency of incorrect responses for an infinite sequence of random discrete\footnotemark{} inputs $\inputvect$ that occur according to the input space probability mass function $\inputspacedist$:  
\begin{IEEEeqnarray}{c}\label{eq:prob-fail}
\prob{f(\inputvect) \neq \outputvect} \definedAs \sum_{\inputvect \in \domain} \ind_{f}(\inputvect)\inputdist(\inputvect)
\end{IEEEeqnarray}

Later (Section~\ref{ss:generalization}), we describe how such long-term failure behavior characterizes \emph{insufficient generalization performance} of an MLM.
\footnotetext{For continuous values, an integration and a probability density function, respectively, replace the summation and the probability mass function. A similar formulation for Eq.~\eqref{eq:prob-fail} is also referenced as \emph{true error} in~\cite{mitchell1997}, \emph{probability of misclassification per random input} for classifiers in~\cite{Zhao2023}, and is the complement of the \emph{probability of a successful prediction} in~\cite{Scheerer2024}.}

\subsection{Performance Metrics and Requirements}\label{ss:perf-reqs-metrix}

Once an MLM has been constructed, quantitative metrics are typically used to evaluate its prediction quality, \ie how well its responses to inputs not previously seen during its training and development, correspond to the functional intent and the required responses. 
Examples of some commonly used metrics include: (for classification problems) \emph{precision}, \emph{recall}, and \emph{F1 score}, %and the \emph{Jaccard index}\footnote{Also commonly known as \emph{intersection over union} (IOU).} 
as well as (for regression problems) \emph{mean absolute error}, and \emph{mean squared error}. 

When the type of the response of an MLM and its containing MLC are the same, then the same set of metrics may be used for each.
For instance, an MLC classifying its inputs using an ensemble of classifiers can be evaluated using the same classification performance metrics as those used for evaluating the individual MLMs in the ensemble. However, the specific values of those metrics for each MLM in the ensemble may differ from the values of the same metrics when applied to the containing MLC. 

For conventional aircraft systems not integrating ML, performance requirements describe specific attributes of functions or systems, such as the type of performance, accuracy, range, fidelity, resolution, and timing behavior~\cite{arp4754}. 

In addition to the above, MLC and MLM performance requirements express the respective desired long-term behaviors, for which a probabilistic formulation may often be appropriate.
More generally, they invoke the associated performance metrics and specify their admissible values.

\emph{Safety-related} performance requirements for an MLC and MLM are those traceable to QSOs, or to higher-level safety requirements, whose violation causes or contributes to a failure condition of the containing (system or aircraft-level) function. 
Non safety-related performance requirements are, equivalently, those whose violation does not cause or contribute to failure conditions. This paper focuses primarily on the former.

\section{Illustrative Example}\label{s:example}

To explain the derivation of safety-related MLC and MLM performance requirements and metrics from an allocated QSO, we consider an illustrative example as shown in Fig.~\ref{f:aebs-bd}---an aircraft emergency braking system (AEBS) adapted from the prior literature~\cite{aebs-dasc2024} as follows: unlike in Fig.~\ref{f:aebs-bd}, the architecture in~\cite{aebs-dasc2024} does not include pre-processing. Also, it treats the post-processing as a part of the emergency braking controller (EBC) functionality, and (implicitly) equates the machine learnt sign detector (MLSD) with the MLC.

\begin{figure}[h]
	\centering
	\includegraphics[width=0.9\textwidth]{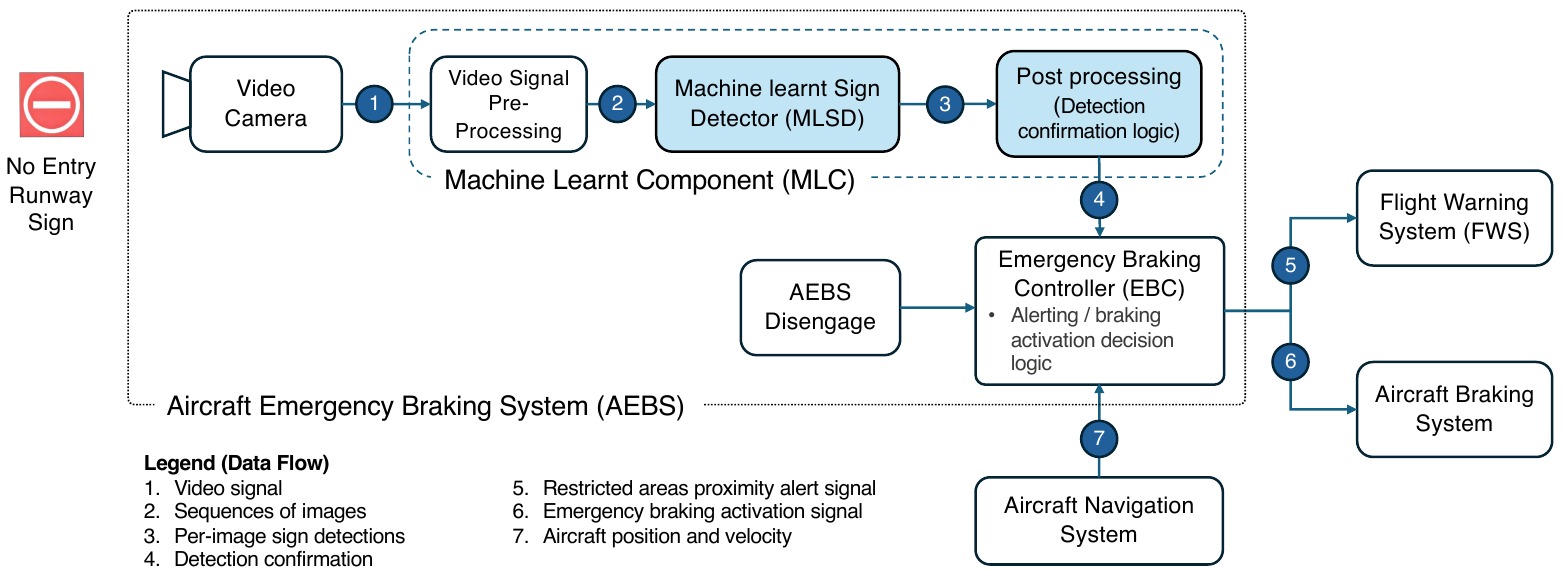}
	\caption{Aircraft Emergency Braking System (AEBS) and its machine learnt component (MLC), adapted from~\cite{aebs-dasc2024}.} 
	\label{f:aebs-bd}
\end{figure}

This section summarizes only those aspects of the AEBS and its safety assessment that are a necessary background for this paper. For more and other details on the AEBS, we refer the reader to~\cite{aebs-dasc2024}.

\subsection{System Description}\label{ss:system-description}

\subsubsection{Functions}\label{ss:function}
The main AEBS function of relevance for this paper is %$\aebsf$: 
\emph{generating an alert} to warn the flight crew (\eg via cockpit annunciation) of the proximity of the aircraft to restricted areas of an airport, which are marked by \emph{No Entry} runway (NER) signs. 
A sub-function of this alerting function allocated to the MLC is %$\mlcf$:
\emph{NER sign detection and classification}. Note that the \emph{emergency braking function} of the AEBS is not in scope for this paper and, as such, affects the safety assessment described later (see Sections~\ref{ss:safety-assessment}, and~\ref{ss:new-pssa})

\subsubsection{Machine Learnt Component}\label{sss:mlc-aebs}

As shown in Fig.~\ref{f:aebs-bd}, The MLC comprises a machine learnt sign detector (MLSD) and its related pre- and post-processing functionality. The MLSD is an implementation of an MLM on target hardware. Here, the MLM is a deep convolutional neural network trained to detect and classify NER signs using supervised, offline learning. 

The MLSD inputs (data flow 2) are sequences of images produced after pre-processing the video signal (data flow 1) from an aircraft mounted, forward facing video camera. 
Video signal pre-processing represents the functionality necessary for the runtime consistency of the types of inputs that the MLSD receives in use, and those on which it is trained offline.
The MLSD responses are a sequence of per image detections or non-detections (data flow 3), corresponding to the input image sequence. Those responses undergo post-processing, a key aspect of which is to confirm or reject confirmation of the detection of an NER sign in a \emph{detection vector}, \ie a fixed size sub-sequence created from the sequence of MLSD responses. 

The confirmation of NER sign detection from the post-processing (data flow 4) is then used by the emergency braking controller (EBC) to send a \emph{restricted areas proximity} (RAP) alert (data flow 5) to the flight warning system (FWS), or an emergency braking activation signal (data flow 6) to the aircraft braking system. As previously mentioned, we do not consider the latter for the rest of this paper.

It is worth noting that the post-processing as shown in Fig.~\ref{f:aebs-bd} is closely coupled to the NER sign detection sub-function. Hence, it is an integral and inseparable element of the MLC. However, in~\cite{aebs-dasc2024} this post-processing is treated as an element of the EBC and referred to as \emph{tracking}, with its failure considered to be the \emph{failure to track NER signs} (also see Section~\ref{ss:safety-assessment}). Although, it is in fact \emph{detection confirmation}, the term we will use henceforth, rather than true tracking.
The detection confirmation logic uses a \emph{confirmation threshold} (\ie a required number of true per image detections in the detection vector) to confirm that an NER sign has indeed been detected when one exists. This confirmation does not require a specific order of detections in the detection vector. 

The detection vector size ($\vectsize = 12$) is determined by:  
\begin{inparaenum}[(i)]
	\item the \emph{detection window period} (the time in which the MLSD must detect an NER sign and raise an alert, so that the aircraft can then be safely decelerated and halted either by the pilot or by automation), and 
	\item the \emph{detection frequency} (the rate at which the MLSD produces per image detections). 
\end{inparaenum}
Those parameters, in turn, depend on various characteristics of the AEBS, the crew, and the aircraft, which include: 
\begin{inparaenum}[(a)]
	\item the maximum taxiing speed (\qty{30}{\knot} $\approx$ \qty{15.43}{\meter\second^{-1}}), 
	\item the maximum deceleration (\qty{6}{\meter\second^{-2}}), 
	\item the pilot reaction time (\qty{3}{\second}), and 
	\item the maximum distance from which a detection is required (\qty{85}{\meter}). 
\end{inparaenum}
We do not repeat the derivation of those parameters, previously detailed in~\cite{aebs-dasc2024}, as it is not required or relevant for this paper. We also note that although the AEBS shown in Fig.~\ref{f:aebs-bd} modifies and adapts the original from~\cite{aebs-dasc2024}, it does not alter those parameter values or their derivation.

\subsection{Safety Assessment}\label{ss:safety-assessment}

The safety effects for which the AEBS is a preventative safety barrier are:
\begin{inparaenum}[(i)]
	\item an inadvertent incursion into a prohibited area, such as a taxiway meant to be used in a given direction; and
	\item an excursion from an aircraft movement surface onto one not meant for aircraft, such as an intersecting roadway. 
\end{inparaenum}

As mentioned earlier (Section~\ref{ss:system-description}), in this paper the scope of the intended use of the AEBS is mainly pilot assistance, even though it includes the capacity for automatic intervention when there is a RAP violation. 
Thus, the primary safety barrier is still piloting procedures in the runway environment, \ie the pilot visually acquires NER signs whilst taxiing, and decelerates upon approaching a restricted area. As such, the AEBS serves as an \emph{additional} protection layer, \eg by providing a RAP alert that will warn the crew if they are distracted. This consideration influences the criticality assigned to the failure conditions of the AEBS function. 

An FHA and PSSA for the AEBS have been given previously in~\cite{aebs-dasc2024}, which we summarize next, to contextualize the rest of the paper.
Specifically, the AEBS functional failure conditions of interest are $\ProximityLoss$: \emph{Loss of RAP alert (crew unaware)}, and $\ProximityMalfn$: \emph{Malfunction of RAP alert}, each of which are assigned a $\mathsf{MINOR}$ severity and a QSO of $10^{-3}~\pfh$, as per the FHA in~\cite{aebs-dasc2024}. 
Additionally, the PSSA invokes a quantitative FTA~\cite{aebs-dasc2024} to relate $\ProximityLoss$ to so-called \emph{ML performance failures}, in particular a \emph{failure of the EBC to track NER signs due to MLC false negatives} allocating to it a QSO of \num{4e-4} per flight. That target is then halved to account for the assumptions of encountering an average of 2 NER signs per flight, and an average flight duration of 4\si{\hour}, resulting in an effective QSO of \num{2e-4} per flight for the MLC.

\newcommand{\flsalrt}{\mathtt{FPProxAlrt}}
\newcommand{\missalrt}{\mathtt{FNProxAlrt}}

\section{Methodology}\label{s:methodology}

\subsection{Assumptions}\label{ss:assumptions}

To simplify the illustration of the proposed method, we assume the following: 
\begin{compactenum}[(1)] 
	\item the camera in the AEBS is functional, operating normally, calibrated, stably mounted, and faithfully captures and transmits the environmental scene as a sequence of images; 
	\item the environmental scene does not contain other signs or objects that could be mistaken as an NER sign; 
	\item there are no transmission errors in the data flow from the video camera through the pre-processing, the MLSD, the post-processing, and the EBC,  to the FWS, so that the data transmitted are uncorrupted and have the correct temporal order as captured by the video camera; and 
	\item pre-processing does not introduce undesired information into the image stream, \eg adversarial transformations.
\end{compactenum}

\begin{figure}[h]
	\centering
	\includegraphics[width=0.575\columnwidth]{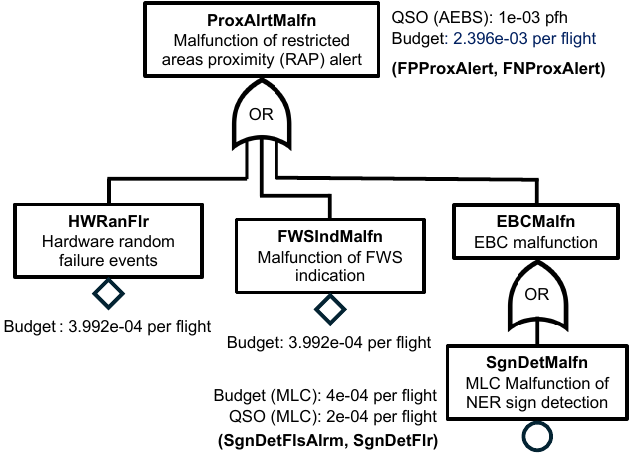}
	\caption{Fault tree relating the malfunction of the RAP alert failure condition of the AEBS, to MLC malfunction.}
	\label{f:aebs-fta}
\end{figure}

\subsection{Revised PSSA and QSO Allocation}\label{ss:new-pssa}

The adaptation of the AEBS (see Section~\ref{sss:mlc}, and Fig.~\ref{f:aebs-bd}) from the original architecture in~\cite{aebs-dasc2024} induces  modifications to the previously mentioned safety assessment. Additionally we identify some corrections to the FTA in~\cite{aebs-dasc2024}. Fig.~\ref{f:aebs-fta} shows a revised fault tree for the failure condition $\ProximityMalfn$, reflecting the following combination of functional failures.

First, $\ProximityMalfn$ can be specialized as two mutually exclusive states: $\flsalrt$: \emph{Inadvertent RAP alert (alert issued when not required)}, and $\missalrt$: \emph{Missing RAP alert (alert not issued when required)}. 
In~\cite{aebs-dasc2024}, only the former has been identified in the FHA as a failure condition, whereas the latter has been incorrectly considered as equivalent to $\ProximityLoss$ in the FTA. Indeed, $\missalrt$ can occur when the AEBS and FWS are both operational and available.

\newcommand{\mlcmalfn}{\mathtt{SgnDetMalfn}}
\newcommand{\flsalrm}{\mathtt{SgnDetFlsAlrm}}

\newcommand{\hwranflr}{\mathtt{HWRanFlr}}
\newcommand{\fwsindmalf}{\mathtt{FWSIndMalfn}}
\newcommand{\ebcmalf}{\mathtt{EBCMalfn}}

Next, from a functional flow standpoint $\ProximityMalfn$ results from a combination of: 
\begin{compactenum}[(i)]
	\item $\fwsindmalf$: \emph{Malfunction of the FWS indication}, 
	\item errors in the FWS alerting logic, or 
	\item $\ebcmalf$: \emph{EBC malfunction}.
\end{compactenum}
$\ebcmalf$ can itself result from errors in the alerting activation decision logic in the EBC, or from the input to the EBC (data flow 4 in Fig.~\ref{f:aebs-bd}), reflected as $\mlcmalfn$: \emph{MLC malfunction of NER sign detection}. 
That, in turn, manifests as one of two mutually exclusive states\footnotemark{}, \ie $\flsalrm$: \emph{False confirmation of an NER sign} (a false positive), and $\failToDetect$: \emph{Failure to confirm detection of the NER sign} (a false negative). 
\footnotetext{In general, sign detection malfunctions are \emph{false positives} or \emph{false classifications} where, for example, either one type of runway sign is misclassified as a different type of sign, or as not a sign (\ie a \emph{false negative}). However, in this example, since the MLM is a binary classifier trained specifically for NER sign detection, the MLC produces a Boolean confirmation response.}

Per the recommended practice~\cite{arp4761}, the fault tree in Fig.~\ref{f:aebs-fta} excludes events corresponding to errors in conventional software, \ie logic errors in the EBC and the FWS. Additionally, it includes $\hwranflr$: \emph{Hardware random failure events}, to aggregate and abstract \emph{other} hardware failures that can also lead to the top event. 
We also include $\mlcmalfn$ in the fault tree, noting that this basic event represents \emph{insufficient MLC performance} rather than a hardware random failure. This is a departure from the conventional practice, justified by the discussion in Section~\ref{sss:io-prob}.

For convenience and comparison to the prior literature, we retain the failure probability budgets and QSOs from~\cite{aebs-dasc2024} for both the top event, $\ProximityMalfn$, and the basic event of the malfunction of the MLC, $\mlcmalfn$, as shown in Fig.~\ref{f:aebs-fta}. It can be easily confirmed that the probability budgets as shown are correct with respect to the fault tree logic.

Thus, as indicated in Section~\ref{ss:safety-assessment}, the effective QSO for $\mlcmalfn$ is \num{2e-4} per taxi operation. Also note that changes to these budgets do not affect the discussion that follows on the proposed method for deriving performance requirements; however the concrete requirements will indeed change.

\subsection{Scope of MLC Behavior}\label{ss:scope}

Again, for convenience, and ready comparison to~\cite{aebs-dasc2024}, in what follows, we mainly focus on the taxiing scenarios where an NER sign is actually present. 
As such, the failure condition $\ProximityMalfn$ effectively presents as the state $\missalrt$ (\ie RAP alert not issued when required), and, likewise, the basic event $\mlcmalfn$ is the state $\failToDetect$ (\ie a failure of the MLC to confirm detection of the NER sign). 
Together with the earlier assumptions (Section~\ref{ss:assumptions}), the scope of MLC behavior and the subsequent analysis for developing safety-related performance requirements for this paper is constrained as follows:
\begin{compactitem}
	\item When the operating environment contains an NER sign, then the responses of the MLSD (see Fig.~\ref{f:aebs-bd}) to an input image containing that NER sign are either: 
	\begin{inparaenum}[(i)]
		\item a \emph{hit}, \ie a correct (true positive) detection of the NER sign (including  correct bounding boxes and class labels), or 
		\item a \emph{miss}, \ie all MLSD responses that are not a hit. Effectively, a miss is only a false negative, since false positives or false classifications cannot be produced in scenarios where an NER sign is actually present in the environment. 
	\end{inparaenum}

	\item Depending on the number of hits and misses determining the confirmation threshold in the detection vector, the detection confirmation logic either confirms an NER sign detection, or it does not confirm an NER sign detection. 

	\item Thus, in all taxiing scenarios where an NER sign is present in the operating environment, when the post-processing does not confirm a sign detection, it represents the occurrence of an MLC malfunction in the state $\failToDetect$, with an effective QSO of \num{2e-4} per flight (taxi operation).

\end{compactitem}

\subsection{From Safety Objectives to Safety-related Performance}\label{ss:intuition}

The QSO allocated to $\mlcmalfn$ is the starting point for deriving the MLC performance requirements and metrics in the AEBS. As clarified above, that event is based on the per image detections received from the MLSD, in the detection vector, during post-processing. 

Specifically, according to the detection confirmation logic, a non-detection occurs when the detection vector contains fewer per image hits than the minimum permissible number of hits required to confirm detection. In other words, when an NER sign is present, to avoid $\failToDetect$: 
\begin{compactenum}[(i)]	
	\item the detection vector must contain at least as many hits as the confirmation threshold; and 
	\item the confirmation threshold should be defined such that the probability of not confirming a detection must be lower than the QSO allocated to $\failToDetect$. 
\end{compactenum}
Note that a related concept of \emph{rejection threshold} can be considered that results in not confirming that an NER sign has been detected. 
Thus, the confirmation (or rejection) threshold is a parameter relevant for safety-related performance.

Additionally, when an NER sign is present in the operating environment and the detection vector contains more per image misses than hits, it suggests that the MLSD has a larger than required \emph{per image probability of non-detection} (equivalently, the \emph{per image miss probability}) leading to the rejection threshold being satisfied. Thus, the per image miss probability  is a safety-related model performance metric, and to avoid $\failToDetect$ it should be defined such that the rejection threshold is not met (or, equivalently, the confirmation threshold is met).

\section{Safety-related Performance Requirements}\label{s:surrogate-model}

We now formalize the preceding intuition as an abstraction of the required behavior (Fig.~\ref{f:aebs-abstraction}), from which we formulate safety-related performance metrics and requirements for the MLC and its underlying MLM. The focus is on specifying requirements rather than verifying that the requirements have been met.

\begin{figure}[h]
	\centering
	\includegraphics[width=0.9\textwidth]{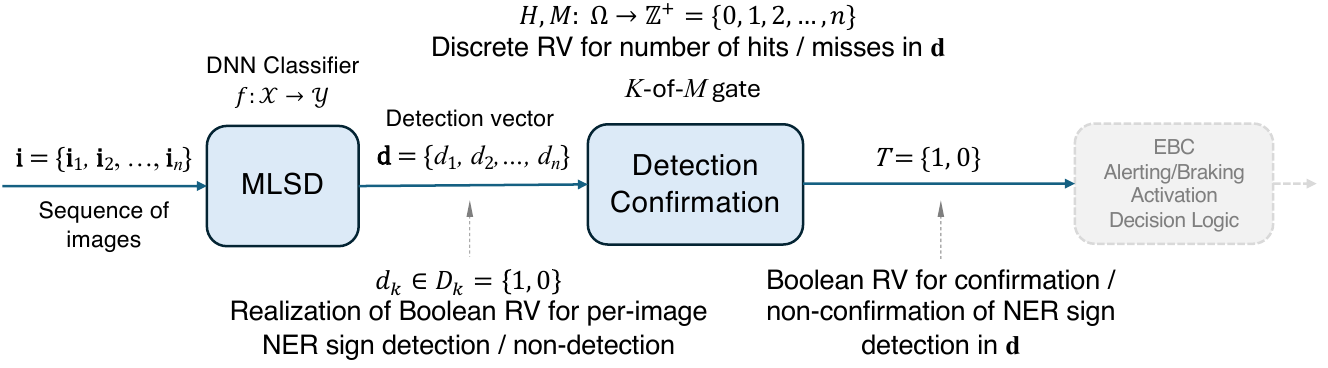}
	\caption{Abstraction to describe the required behavior for NER sign detection using the MLC.}
	\label{f:aebs-abstraction}
\end{figure}

\subsection{Abstraction of Required Behavior}\label{ss:abstraction}

Let $T = \{0,1\}$ be a Boolean random variable (RV) for the event of an MLC response, the output of the post-processing detection confirmation logic. 
Those responses are either a confirmation of detection of an NER sign, \ie the event $(T=1)$, or the malfunction $\mlcmalfn$, \ie the event $(T=0)$.  
As clarified in Sections~\ref{ss:new-pssa} and~\ref{ss:scope}, $\mlcmalfn$ occurs as the state $\failToDetect$, \ie a failure to confirm detection of the NER sign. Hence, 
\begin{IEEEeqnarray*}{c}
	\failToDetect \definedAs (T=0)
\end{IEEEeqnarray*}
Let the QSO allocated to $\mlcmalfn$ be $\qso$. Thus, a concrete safety-related MLC performance requirement for NER sign detection, based on the allocation from the PSSA process (specifically, the FTA in Fig.~\ref{f:aebs-fta}), is:

\begin{crequirement}[MLC Safety Performance]\label{req:fail-to-detect}
The average probability of non-detection of an NER sign per taxi operation shall be less than $\qso$, \ie $\prob{T=0} < \qso \leftarrow \num{2e-04}$
\end{crequirement}

We can specify an analogous requirement on the MLC \emph{functional performance} as: 

\begin{crequirement}[MLC Functional Performance]\label{req:detect-success}
The probability of detecting an NER sign shall be at least $(1-\qso)$, \ie $\prob{T=1} \geq (1-\qso) \leftarrow \num{0.9998}$
\end{crequirement}

From Fig.~\ref{f:aebs-abstraction}, the detection vector, $\detvect = \{d_{1}, d_{2}, \ldots, d_{n}\}$, of size $\vectsize$, is a finite sequence of responses produced by the MLSD, $f$, to a sequence of input images $\{\mathbf{i}_{j}\}_{j=1}^{n}$. 
Here, $d_{j} \in \{1, 0\}$ is the realization of $D_{j}$, a Boolean RV representing the event of the $j^{\textrm{th}}$ response of $f$ to the $j^{\textrm{th}}$ input image $\mathbf{i}_{j}$. If $(D_{j} = 1)$, $(D_{j} = 0)$ represent a hit and a miss, respectively, then whenever there is a hit in $\detvect$, $d_{j}=1$, otherwise $d_{j} = 0$. 

Let the confirmation and rejection thresholds be $\minhit$ and $\minmiss$, respectively. Also let $H, M$ be the discrete RVs for the number of hits and misses, respectively, whose realizations are $h, m \in \{0, 1, 2, \ldots, n\}$.
As clarified earlier (Section~\ref{sss:mlc} and~\ref{ss:intuition}), the post-processing confirms that an NER sign has been detected when $h \geq \minhit$. 
Moreover, since hits do not need to occur in a specific order in $\detvect$ for a detection confirmation, the corresponding logic is a $K$-of-$M$ \emph{gate}, where $K = \minhit$ and $M = n$. 

We can now readily confirm that $h$ is the sum of the individual detections in $\detvect$, and formalize the detection confirmation logic as: $\forall \detvect, (h = \sum_{j=1}^{n}d_{j}) \geq \minhit \Rightarrow (T=1)$. We will concretize this as a requirement next, in Section~\ref{ss:concrete-reqs}.

Since $\detvect$ contains a combination of hits and misses, we have $\vectsize = h + m$, and when $h = \minhit$, then $m = (\vectsize - \minhit)$ represents the maximum permissible per image misses in $\detvect$ that still results in a detection confirmation. Hence, an additional miss will result in a failure to confirm detection, so that $\minmiss = (\vectsize - \minhit) + 1$. 

Now, assume that a hit or miss response of $f$ is the result of a Bernoulli trial, and that each $D_{j} \in \detvect$ is independent and identically distributed (IID)\footnote{Section~\ref{ss:threats-to-validity} justifies these assumptions and discusses their implications.}. 
Then, let the \emph{per image hit probability}, $\prob{D_{j} = 1} = \hitprob$, so that the \emph{per image miss probability}, $\prob{D_{j} = 0} = \missprob = 1 - \hitprob$. 

We have that $\detvect$ is a realization of a \emph{Bernoulli process}, \ie a sequence formed by the result of $\vectsize$ Bernoulli trials in which there are $h$ events such that $(D_{j} = 1)$ and $m$ events such that $(D_{j} = 0)$. 
Since the sum of the RVs of a Bernoulli process is another RV that follows a binomial distribution, 
$H \sim \mathtt{Binomial}(n, \hitprob)$, and the probability of at least $h$ hits is 
\begin{IEEEeqnarray}{c}\label{eq:atleast-h}
  \prob{H \geq h} = \sum_{i=h}^{n}\binom{n}{i} \hitprob^i (1-\hitprob)^{n-i}
\end{IEEEeqnarray} 
Hence, the probability of confirming an NER sign detection is
\begin{IEEEeqnarray}{c}\label{eq:det-confirmation-prob}
	\prob{T=1} =	\prob{H \geq \minhit} = \sum_{i=\minhit}^{n} \binom{n}{i} \hitprob^{i} (1-\hitprob)^{n-1}
\end{IEEEeqnarray}

Then, the probability of failure to confirm detection of the NER sign, $\prob{T=0}$, is $1 - \prob{T=1}$, which we formulate in terms of $M$, $n$, $\minmiss$, and $\missprob$. That is, 
\begin{align}\label{eq:det-fail-prob}
		\prob{T=0} = \prob{M \geq \minmiss} = \sum_{i=\minmiss}^{n} 
					 		\binom{n}{i}\missprob^{i} (1-\missprob)^{n-1}
\end{align}

\subsection{Concrete Performance Requirements}
\label{ss:concrete-reqs}

To establish concrete requirements for $\hitprob$, $\missprob$, $\minhit$ and $\minmiss$, we solve %~\eqref{eq:det-confirmation-prob} to satisfy Req.~\ref{req:fail-to-detect}, noting that $n=12$ (see Section~\ref{sss:mlc-aebs}).
either of \eqref{eq:det-confirmation-prob} and~\eqref{eq:det-fail-prob} such that reqs.~\ref{req:fail-to-detect} or~\ref{req:detect-success}, respectively, are satisfied. 
%%
%%Fig.~\ref{f:non-detect-vs-miss} and Fig.~\ref{f:detect-vs-hit} show that the solution is not unique. 
%
%%In particular, 
%
Fig.~\ref{f:non-detect-vs-miss} %graphs 
shows a graphical solution, varying $\prob{T=0}$ on a logarithmic scale, for the rejection thresholds $12 \geq \minmiss > m \in [4, 11]$, and a range of $\missprob = [0, 0.5]$. 

\begin{figure}[h]
	\centering
		\includegraphics[width=\textwidth]{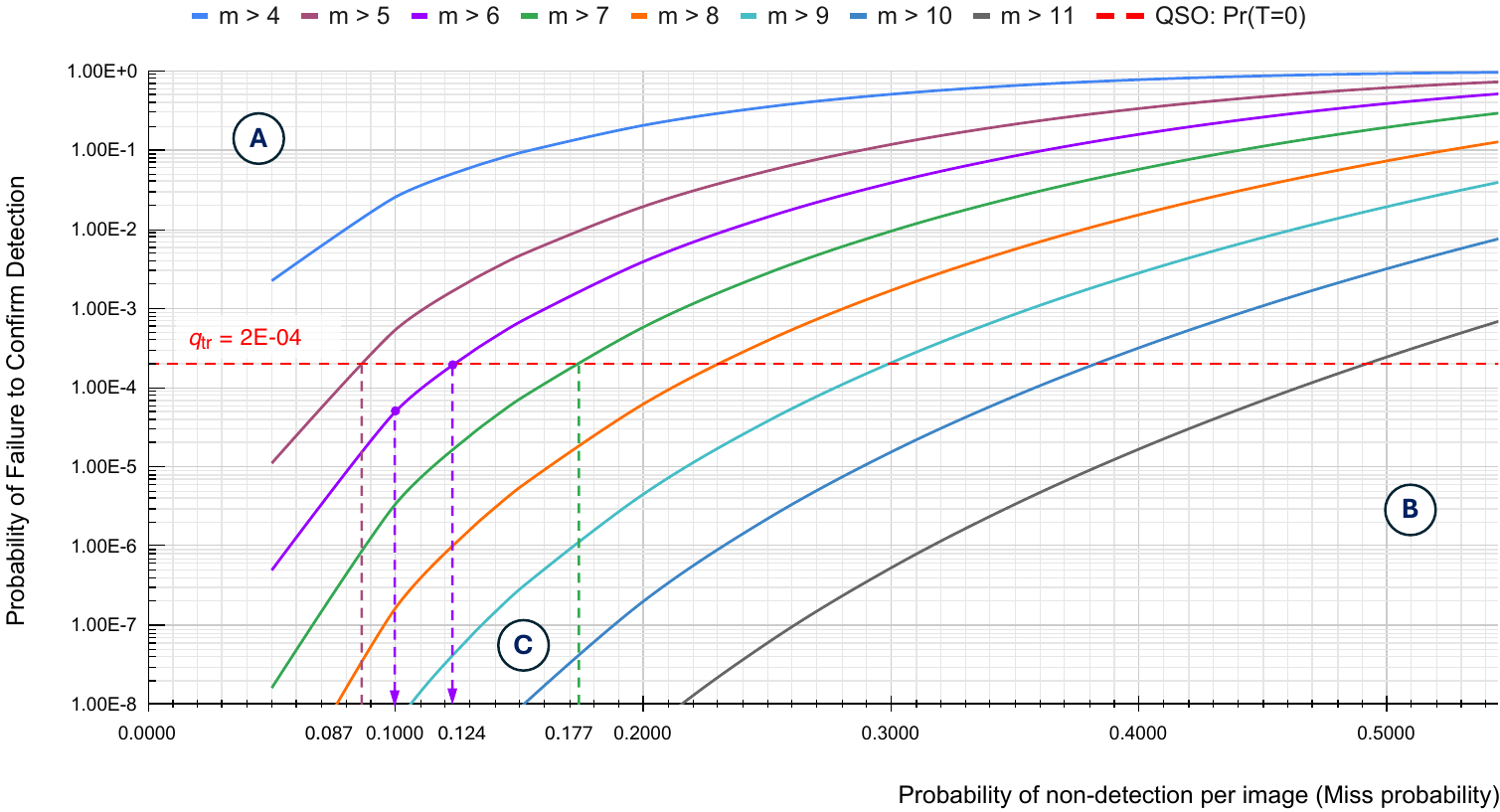}
		\caption{Varying $\prob{T=0}$ on the $y$-axis, with $\missprob$ on the $x$-axis, for different values of $m$, determined from~\eqref{eq:det-confirmation-prob}.}
		\label{f:non-detect-vs-miss}
\end{figure}

The dotted horizontal line in Fig.~\ref{f:non-detect-vs-miss} is the QSO for failing to confirm NER sign detection. As shown, the QSO is not met in region $\mathsf{A}$, but is satisfied in region in $\mathsf{B}$ for rejection thresholds $\minmiss = 12 \geq m > 5$. 
The region $\mathsf{C}$, between the two vertical dotted lines, is a sub-region of $\mathsf{B}$, giving a candidate range for %of miss probabilities, 
$\missprob \approx [0.087, 0.177]$ and %rejection thresholds, 
$\minmiss = [6,8]$, respectively.
Then, together with the previous discussion (Section~\ref{ss:abstraction}), we  obtain a range for $\hitprob \approx [0.823, 0.913]$ and $\minhit = [5, 7]$.

%Similarly Fig.~\ref{f:detect-vs-hit} graphs the variation of $\prob{T=1}$  (again, on a logarithmic scale), for a range of $\hitprob = [0.75, 0.95]$, but now considering the confirmation thresholds resulting from the range for $\minmiss$ previously selected, \ie $\minhit \geq h \in [5, 7]$.
%%
%As before, the dotted horizontal line in Fig.~\ref{f:detect-vs-hit} is the \emph{functional performance objective}, ($1-\qso$) for sign detection confirmation, ($T=1$). That objective is met in region $\mathsf{A}$, but not met in region $\mathsf{B}$, for the chosen $\minhit$. 
%%
%Again, the region $\mathsf{C}$, between the two dotted vertical lines is a sub-region of $\mathsf{A}$ giving a range for $\hitprob \approx [0.823, 0.913]$ and $\minhit = [5, 7]$. 

\newcommand{\missratio}{m_{\mathrm{t}}}
\newcommand{\meanm}{\mu_{m}}
\newcommand{\stddevm}{\sigma_{m}}

%From Fig.~\ref{f:non-detect-vs-miss} and Fig.~\ref{f:detect-vs-hit}, we can select, say, $\minhit = 6$ and $\hitprob = 0.9$; therefore, $\minmiss = 7$, and $\missprob = 0.1$.   

From Fig.~\ref{f:non-detect-vs-miss}, we have $\missprob' \approx 0.124$, where the QSO is exactly $\qso$ for $\minhit = 6$. 
We can now select, say, $\minmiss > 6$, and $\missprob = 0.1$, so that $\minhit = 6$ and $\hitprob = 0.9$. 
Then we can specify additional concrete performance requirements for the MLC and its elements, namely the MLSD, and its post-processing. 

First, the concrete requirement based on the formalization of the detection confirmation logic is:

\begin{crequirement}[MLC Detection Confirmation]\label{req:detection-concrete}
The MLC post-processing shall confirm an NER sign detection whenever there are at least $6$ detections in any detection vector, \ie $\forall \detvect, \minhit \geq 6 \Rightarrow (T=1)$
\end{crequirement}

We may equivalently specify the dual of Req.~\ref{req:detection-concrete}  specifying the rejection of sign detection confirmation based on the rejection threshold as:

\begin{crequirement}[MLC Reject Detection Confirmation]\label{req:detection-rejection-concrete}
The MLC post-processing shall reject confirmation of an NER sign detection whenever there are at least $7$ non-detections in any detection vector, \ie $\forall \detvect, \minmiss \geq 7 \Rightarrow (T=0)$
\end{crequirement}

Then, similar to Req.~\ref{req:fail-to-detect}, %and~\ref{req:detect-success}, 
the required MLSD safety performance in terms of the respective miss %/hit 
probability is:  

\begin{crequirement}[MLSD Safety Performance]\label{req:mlsd-safe-perf}
The MLSD shall have a per image probability of non-detection of an NER sign of at most $0.1$, \ie $\forall d_{j} \in \detvect; \prob{D_{j} = 0} = \missprob \leq 0.1$
\end{crequirement}

As earlier, we can give an analogous requirement for MLSD functional performance %(again, not given here due to space constraints) 
in terms of $\hitprob$ 
as: %$\forall d_{j} \in \detvect; \prob{D_{j} = 1} = \hitprob \geq 0.9$
\begin{crequirement}[MLSD Functional Performance]\label{req:mlsd-func-perf}
The MLSD shall have a per image probability of detection of an NER sign of at least $0.9$, \ie $\forall d_{j} \in \detvect; \prob{D_{j} = 1} = \hitprob \geq 0.9$
\end{crequirement}

We additionally specify MLSD safety performance in terms of a \emph{tolerable miss ratio} metric, $\missratio$, \ie the allowable proportion of missed detections per detection vector, which we compute as: $\missratio = (\meanm + \stddevm)/n  \approx 0.187$, where $\meanm = n\missprob$ is the mean, and $\stddevm^{2} = n\missprob(1-\missprob)$ is the variance, respectively, of $M\sim \mathtt{Binomial}(n, \missprob)$. 
Hence: 
%
%\begin{crequirement}[MLC Functional Performance -- Average Hit Count]\label{req:hitrate}
%The average hit count for the MLC shall be in the range $[9.76, 11.83]$, \ie $\minhit/\vectsize \geq 0.5$
%\end{crequirement}
%
\begin{crequirement}[MLC Safety Performance -- Miss Ratio]\label{req:missrate}
The tolerable miss ratio for the MLSD shall not exceed $0.187$, \ie $\missratio \leq 0.187$
\end{crequirement}

\subsection{Generalization Performance Requirements}\label{ss:generalization}

\newcommand{\rpop}{R_{p}}
\newcommand{\remp}{R_{\mathrm{e}}}
\newcommand{\loss}{\boldsymbol\ell}
\newcommand{\expect}{\mathbb{E}}
\newcommand{\data}{\mathcal{D}}
\newcommand{\trainingdata}{\mathcal{D}_{\mathrm{train}}}
\newcommand{\testdata}{\mathcal{D}_{\mathrm{test}}}
\newcommand{\fnr}{\mathtt{FNR}}
\newcommand{\tpr}{\mathtt{TPR}}
\newcommand{\jointdist}{\mathrm{Pr}_{\domain,\codomain}}

Recalling the discussion in Section~\ref{sss:io-prob}, when the MLM failure probability exceeds its requirement as derived from the safety objectives, its generalization performance is insufficient. In other words, the minimum required generalization performance requirement is related to the maximum tolerable MLM failure probability, which we now characterize in terms of the model \emph{generalization error} and \emph{generalization gap}.

\subsubsection{Generalization Error}\label{sss:gen-error}

The generalization error $\rpop(f)$ for an MLM $f$, also known as the \emph{population risk}, is defined as the expected value of a suitable \emph{loss function}, $\loss\left(f(\inputvect),\outputvect\right)$, evaluated over the joint distribution $\mathrm{Pr}_{\domain,\codomain}(x,y)$ of the input and output spaces for $f$. Thus,  
\begin{IEEEeqnarray}{c}\label{eq:pop-risk}
	\rpop(f) \definedAs \expect_{(x,y)\sim \mathrm{Pr}_{\domain,\codomain}} 	
	\loss\left(f(\inputvect), \outputvect\right)
\end{IEEEeqnarray}

For binary classification, a commonly used loss function is the so-called \emph{zero-one} loss, defined as
\begin{IEEEeqnarray}{c}\label{eq:zero-one-loss}
	\loss\left(f(\inputvect, \outputvect)\right) \definedAs 
		\ind_{f}(x) = 
			\begin{cases}
				1 & \text{when } f(\inputvect) \neq \outputvect \\
				0 & \text{otherwise}
			\end{cases}
\end{IEEEeqnarray}

\begin{theorem}\label{thm:gen-error-is-fail-prob}
	The generalization error for an MLM performing binary classification is exactly its failure probability under the zero-one loss.
\end{theorem} 

\begin{proof}
%
%Substituting~\eqref{eq:zero-one-loss} into~\eqref{eq:pop-risk}, and then, 
%from the definition of expectation, by rearranging terms, marginalization over $y$, and Eq.~\eqref{eq:prob-fail}, respectively,
%
We have 
\begin{IEEEeqnarray}{rCl's}\label{eq:pop-risk-is-prob-fail}
%\begin{IEEEeqnarray}{rCl}\label{eq:pop-risk-is-prob-fail}
	\rpop(f) & = & \expect_{(x,y)\sim \mathrm{Pr}_{\domain,\codomain}} 
									\ind_{f}(x) %\IEEEyessubnumber \\
							 & \ldots By substituting~\eqref{eq:zero-one-loss} 
							 		into~\eqref{eq:pop-risk} \IEEEyessubnumber\\ 
					 & = & \sum_{x,y}\ind_{f}(x)\jointdist(x,y) %\IEEEyessubnumber\\ 
					 		 & \ldots From the definition of expectation \IEEEyessubnumber\\ 
					 & = & \sum_{x} \ind_{f}(x) \sum_{y} 
								\jointdist(x, y) %\IEEEyessubnumber\\ 
							 & \ldots Distributive property \IEEEyessubnumber \\ 
					 & = & \sum_{x} \ind_{f}(x) \inputdist(x) %\IEEEyessubnumber\\ 
					 		 & \ldots By marginalization over y \IEEEyessubnumber \\
					 & = & \prob{f(\inputvect) \neq \outputvect}
					 	   & \ldots From~\eqref{eq:prob-fail}
\end{IEEEeqnarray}
\end{proof}

%Thus, the generalization error is exactly the failure probability under the zero-one loss. 

Now, if $\inputvect$ is an input image containing an NER sign, $\mathbf{i}_{j}$, then as clarified in Section~\ref{ss:abstraction},  the required MLSD response $\outputvect$ is $d_{j} = 1 $ in $\detvect$, \ie $(D_{j} = 1)$; hence,~\eqref{eq:pop-risk-is-prob-fail} is equivalent to the per image probability of non-detection of an NER sign, $\prob{D_{j} = 0}$. From Fig.~\ref{f:non-detect-vs-miss}, $\prob{D_{j} = 0}$ attains its maximum tolerable value when $\qso$ is met; therefore $\rpop(f) = \missprob'$. Thus, 

\begin{crequirement}[MLSD Generalization Performance]\label{req:gen-perf}
	The MLSD shall have a generalization error of at most $0.124$, \ie 
	$\rpop(f) \leq \missprob' \leftarrow 0.124$
\end{crequirement}

However, neither the joint nor the input distribution may be exactly known. Thus, although we can require $\rpop$ to be $\missprob'$, its \emph{true} value cannot be determined. Instead, in practice, $\rpop(f)$ is estimated using the \emph{empirical test risk} metric, $\remp(f,\testdata)$, under the requirement that the test data $\testdata$ (as well as the training data) used to learn $f$ are sampled from a \emph{representative} joint distribution. 
For the zero-one loss the empirical risk measured on dataset $\data$ is, in fact, the \emph{false classification rate} performance metric~\cite{murphy-pml}. Recalling Section~\ref{ss:scope}, the false classification rate for $f$ is the \emph{false negative rate}, $\fnr(f, \data)$. 
Thus, we can refine Reqs.~\ref{req:mlsd-safe-perf} and~\ref{req:gen-perf} as: 
%and~\ref{req:mlsd-func-perf} can be 

\begin{crequirement}[MLSD Performance -- False Negative Rate in Test]\label{req:fnr}
	The MLSD shall have a false negative rate in test of at most $0.1$, \ie 
	$\fnr(f, \testdata) \leq \missprob \leftarrow 0.1$ 
\end{crequirement}

Additionally, the \emph{true positive rate} performance metric (also known as \emph{sensitivity} or \emph{recall}), $\tpr(f, \data)$, measured on a dataset $\data$, is the dual of the false negative rate. Thus, 
%we can also give an analogous requirement for the recall in test as:

\begin{crequirement}[MLSD Performance -- Recall in Test]\label{req:tpr}
	The MLSD shall have a recall in test of at least $0.9$, \ie 
	$\tpr(f, \testdata) \geq (1 - \missprob) \leftarrow 0.9$ 
\end{crequirement}

\subsubsection{Generalization Gap}\label{sss:gen-gap}

\newcommand{\safetymargin}{\mathtt{S_{M}}}

The \emph{empirical training risk}, $\remp(f,\trainingdata)$, is an analogous metric to the empirical test risk. %(Section~\ref{ss:generalization}). 
The difference between the two gives an estimate of the \emph{generalization gap}, which is, itself, the difference between the generalization error and the empirical training risk, \ie 
\begin{IEEEeqnarray}{c}\label{eq:gen-gap}
	\rpop(f) - \remp(f, \trainingdata) \approx \remp(f, \testdata) - \remp(f,\trainingdata)
\end{IEEEeqnarray}

We can now give a probabilistic upper bound $\delta$ to the generalization gap (or to its estimate) using \emph{Hoeffding's inequality} and the \emph{union bound} theorems~\cite{murphy-pml}. Thus, for data $\data$ comprising $\eta$ samples and a tolerance $\epsilon$ in the generalization gap, we have:  
\begin{IEEEeqnarray}{c}\label{eq:hoeffding}
	\prob{\abs{\rpop(f)-\remp(f,\data)} > \epsilon} \leq \delta = 
				2e^{-2\eta\epsilon^{2}}
\end{IEEEeqnarray}
which can also be rearranged as: 
\begin{IEEEeqnarray}{c}\label{eq:samplesize}
	\eta \geq \frac{1}{2\epsilon^{2}}\ln\left(\frac{2}{\delta}\right)
\end{IEEEeqnarray}

Note that some of the available literature refers to $\epsilon$ as \emph{accuracy}, and $\delta$ as \emph{confidence}. To avoid a misinterpretation of those terms as used in the contexts of aircraft certification, and system safety, versus ML, we refer to $\epsilon$ as the \emph{tolerance} and to $\delta$ as the probabilistic upper bound instead.

The minimum number of independent samples required to satisfy the bound can be determined from~\eqref{eq:samplesize}, by selecting the desired tolerance and the probabilistic upper bound.  
For example, select: 
\begin{inparaenum}[(i)]
	\item $\delta = \num{1e-3}$, proportional to the order of magnitude of the QSO, and 
	\item $\epsilon = \safetymargin(\missprob' - \missprob)$, where %$\missprob'$ is as described in Section~\ref{ss:concrete-reqs}, 
%	
%	the per image miss probability at which the QSO is exactly met for the chosen rejection threshold, $\minmiss$, and 
%
$\safetymargin$ is a margin of safety. The reasoning here is that a tolerance greater than the difference in the required generalization error and the required false negative rate, $(\missprob' - \missprob)$, results in a failed detection confirmation.
\end{inparaenum} 
Thus, selecting $\safetymargin = 0.5$, and from Fig.~\ref{f:non-detect-vs-miss}, $\missprob' \approx 0.124 \Rightarrow \epsilon = 0.012$, therefore $\eta \geq 26393$ independent samples (drawn from a representative distribution). 

Depending on whether this procedure is applied to the generalization gap or to its estimate, we can upper bound either of the two and derive the sample sizes of the training and test datasets required at the chosen tolerance. Thus, additional testing-related requirements can then be specified (not given here).

\section{Discussion}\label{s:discussion}

\subsection{Rationale for Assurance of Validity}\label{ss:rationale} 

% \comment[R3]{Not a research design. Yet making an argument that the approach is valid.}

A robust validation of the proposed method and the consequent performance requirements (Sections~\ref{s:methodology} and~\ref{s:surrogate-model}) requires a careful research design, which is out of scope for this paper, and an avenue for future work.
Instead, this section provides rationale to justify why the proposed method and the resulting requirements are a valid step to relate system-level QSOs and the performance requirements of machine learnt functionality. 

%Our assurance argument is mainly that the QSO allocated to the MLC based on the system architecture, and the requirements subsequently formulated from that allocated QSO are each the tolerable worst-case, derived mathematically from the system-level QSO; therefore they represent the minimum requirements.
%
%Each aspect of this reasoning is discussed next. 

\subsubsection{Suitability of the System Architecture and QSOs}
\label{sss:qso-suitable}

The MLC and its organization in the AEBS (Fig.~\ref{f:aebs-bd}) represent a \emph{single channel architecture}. That is, a loss or malfunction of any element of the channel compromises the entire channel. Hence it is the weakest from the standpoint of both reliability and safety.
When decomposing and allocating the QSO to be achieved by such an architecture to its elements (including an MLC), the allocated QSOs are more conservative than they would be for alternative architectures, \eg with redundancy, or diversity. In that sense, given the intended use and the safety assessment (Section~\ref{s:example}), the chosen architecture and the QSO for the MLC are the tolerable worst-case. Therefore they are appropriate and sufficient as a starting point to formulate a conservative set of MLC performance requirements. 

\subsubsection{Suitability of the Performance Requirements}
\label{sss:reqs-suitable}

We model the probability of failure of the MLSD as the limiting relative frequency of incorrect responses to random image inputs from the input space, as given by Eqs.~\eqref{eq:prob-fail}, \eqref{eq:pop-risk} -- \eqref{eq:pop-risk-is-prob-fail}. As such, MLSD failure behavior, as modeled, is equivalent to random failure. 

Then, in the FTA for the AEBS (Fig.~\ref{f:aebs-fta}), we capture MLC malfunction as the EBC malfunction basic event, computing its failure probability as in %Eq.~\eqref{eq:det-fail-prob}. 
Section~\ref{ss:abstraction}. 
This is analogous to the result of a quantitative FTA for a $K$-of-$M$ gate (also known as a \emph{voting gate}), whose basic events are each of the per image responses of the MLSD in the detection vector. Here, a per image non-detection, \ie an incorrect response, is equivalent to the random failure of the corresponding basic event with a constant failure probability $\missprob$. 
Thus, the binomial model for detection confirmation (Section~\ref{ss:abstraction} and Fig.~\ref{f:aebs-abstraction}) abstracts MLC malfunction also as a random failure. 

Now, as clarified in Sections~\ref{sss:deterministic} and~\ref{sss:systematic}, the MLSD is both deterministic and systematic in its behavior. Furthermore, the MLSD implements a deep convolutional neural network (Section~\ref{sss:mlc-aebs}), which has a \emph{feedforward} neural architecture, \ie there are no feedback loops between its neurons in the network layers. Thus, it is also \emph{stateless}, with the responses depending only on the current inputs, and not on the history of inputs or prior responses. %These imply that the MLSD has a \emph{predictable} behavior, \ie the future response is exactly known if the input, the model response to the input, and the nature of the response (correct, incorrect) is known. 
Furthermore, detection confirmation is a deterministic rule-based decision.
Together, it implies that, under the stated assumptions and scope (Sections~\ref{ss:assumptions} and~\ref{ss:scope}), for any input, and input sequence subsequently formed: 
\begin{inparaenum}[(i)] 
	\item the ideal (best case) MLC behavior is a systematically correct response due to perfect generalization of the MLSD and a deterministic choice of NER sign detection confirmation; and  
	\item the worst case MLC behavior is a systematically incorrect response due to consistent MLSD failure followed by a deterministic choice rejecting NER sign detection confirmation. 
\end{inparaenum}

\subsubsection{Validity} 

Since random behavior lies between the worst case and ideal behavior, and since the concrete performance requirements defined based on random behavior (Section~\ref{ss:concrete-reqs}) have been mathematically derived from the  allocated QSO, we can conclude that: 
\begin{inparaenum}[(i)] 
	\item the performance requirements as specified meet the allocated QSO by construction; 
	\item any systematic MLC behavior up to the ideal, verified to meet or exceed the specified requirements will also meet the allocated QSO; and
	\item the requirements as defined are the minimum required, being the tolerable worst-case. 
\end{inparaenum}
%
%Hence the proposed method and the performance requirements that result are valid.

\subsection{Threats to Validity}\label{ss:threats-to-validity}

First, correct application of the FTA for QSO allocation (Fig.~\ref{f:aebs-fta}) may potentially challenge the rationale for the suitability of the allocated QSO (Section~\ref{sss:qso-suitable}). Specifically, the recommended practice for FTA~\cite{arp4761} only considers hardware failure basic events, with associated failure rates, rather than basic events representing insufficient MLC performance, \eg EBC malfunction with a failure probability. 
%However, we also note that it is not uncommon to conflate a \emph{failure rate} with an (average) probability of failure normalized by exposure, and 

However, quantitative FTA admits computation with failure probabilities, and Section~\ref{sss:io-prob} clarifies why a probability of failure can indeed be assigned to the MLC malfunction basic event, thus justifying its inclusion in the FTA. Additionally, note that the purpose of the FTA as in Section~\ref{ss:new-pssa} is to \emph{specify} requirements rather than to \emph{verify} that they have been met. As such, we contend that the way we have applied FTA is sound. 
Furthermore, although changes to the specific value of the system level QSO can change both the allocated QSO and the concrete performance requirements, the proposed method to develop those requirements, itself, is unaffected. 

Next, the constraints for applying a binomial model may potentially challenge its use and the associated rationale for the suitability of the resulting performance requirements (Section~\ref{sss:reqs-suitable}). We enumerate and substantiate each constraint: 
\begin{compactenum}[(i)]
	\item \emph{Fixed number of Bernoulli trials}: Met due to a fixed size detection vector ($n=12$), and by definition (Section~\ref{ss:abstraction}), with Boolean responses for both the MLSD and the MLC.
	\item \emph{IID trials}: The input to the MLC, and subsequently to the MLSD, is a temporally ordered sequence of images of the runway scene as captured and transmitted by the video camera. Hence, they are a correlated time series from the same generating process, due to which they are \emph{identically distributed} but \emph{not independent}. The MLSD is systematic, deterministic, and stateless; hence its responses are also identically distributed but not independent. Since the MLSD responses are the inputs to the binomial model, the IID constraint is \emph{not met}. 
	
\quad However, this constraint effectively implies that the trials should be random. Thus, our counter argument here is that maintaining the assumption of IID trials does not invalidate the requirements because: 
\begin{inparaenum}[(a)]
	\item despite abstracting the MLC failure behavior as random, the concrete performance requirements meet the QSO by construction; and 
	\item as before, we use the abstraction to define the requirements rather than to verify that they have been met, which is when the IID constraint would apply.
\end{inparaenum}
That is, the MLC must be verified with non-IID data against Requirements~\ref{req:fail-to-detect}--\ref{req:tpr}, even though those requirements have been defined assuming IID inputs for post-processing. 

	\item \emph{Constant probability of trial outcomes}: Eqs.~\eqref{eq:prob-fail}, \eqref{eq:pop-risk} -- \eqref{eq:pop-risk-is-prob-fail} clarify the relationship of the MLSD failure probability to the distribution of its inputs, showing that the former is deterministically related to the latter by the identity function. Since the inputs to the MLSD have been established to be identically distributed, their moments (\eg their expected values) are also identical and therefore constant (but unknown). Hence the MLSD failure probability is constant. 

\quad A concern here is that $\missprob$ may change over the long run, due to a drift in the input space distribution. Mitigating the effects of such distribution drift requires carefully describing the ODD (see Section~\ref{sss:odd}), and consideration of the exposure duration over which the QSO and failure probabilities are expressed, \ie the duration for which the input distribution is expected to be stable, and where $\missprob$ will then be constant.
\end{compactenum}

\subsection{Additional Considerations}\label{ss:additional-considerations}

%In addition to the assumptions and constraints on the scope of MLC behavior (Sections~\ref{ss:assumptions} and~\ref{ss:scope}), the following additional observations are noteworthy: 
%\begin{asparaenum}[(i)]
%	\item 

\subsubsection{Relevance of the Operational Design Domain}
\label{sss:odd}

Defining $\jointdist(x,y)$, the joint distribution of the input and output space, underpins both the ML process and the development of MLC and MLM performance requirements. That induces specific additional considerations on sufficiently characterizing: 
\begin{inparaenum}[(i)]
	\item the \emph{marginal} input space distribution, $\inputdist({\domain})$, reflecting the intended operating environment; and 
	\item the \emph{conditional} input space distribution, $\conditionaldist({\domain|\codomain})$, %or equivalently the \emph{likelihood of the responses} $\mathcal{L}(\codomain|\domain)$, 
	which reflects functional intent. 
\end{inparaenum} 
Both considerations require, in part, a well-defined and validated ODD from which data must be sampled to meet various data properties~\cite{easa-cp-L12, ml-nasa-tr-2024, kape-safecomp-2023}. The latter consideration in particular levies requirements on the pre-processing element, or more generally on the system architecture, to assure that the MLC only receives inputs consistent with its defined input space (known as \emph{in-ODD} or \emph{in-domain}) and functional intent (\ie \emph{in-distribution}), as considered during the ML process. 
	
%\comment[R2]{Moreover - if the work has covered this aspect - it may be helpful to show at what level (system, subsystem, component, software) the performance requirements for this functionality were specified, and why, and to note the level of coverage of these requirements - for example, what is the approach for robustness?}
%

\newcommand{\rpopgen}{R_{p}^{(g)}}
\newcommand{\rpoprob}{R_{p}^{(r)}}

\subsubsection{Robustness Performance}\label{sss:robustness}

The ML literature treats model robustness separately from generalization performance. However, from a safety standpoint, we contend that MLM failures in general, especially those that result from model fragility under input perturbation or abnormality, stem from an inadequate definition of both the ODD and the corresponding input space distribution. This viewpoint is consistent with how, for example, a lack of robustness in conventional airborne software is treated as a requirements inadequacy~\cite{do178}. As such, not only normal range inputs but also aberrant and limiting inputs should be considered in the ODD and the corresponding input space distribution when specifying and evaluating MLM failure probability and the associated performance requirements. 

Thus, in this paper, although MLSD robustness performance has not been considered, we indicate a potential way forward for future work: as clarified in Section~\ref{ss:generalization}, the MLSD generalization error $\rpop(f)$ equals its maximum tolerable failure probability $\missprob'$. We propose to treat it as a metric of \emph{robust generalization} that considers failures due to both a lack of model robustness and inadequate generalization for previously unseen inputs. Hence, we can have: 
\begin{IEEEeqnarray*}{c}
	\rpop(f) = \rpoprob(f) + \rpopgen(f) 
\end{IEEEeqnarray*}
where $\rpoprob(f)$ is the portion of the generalization error apportioned to robustness related failures, and $\rpopgen(f)$ is the remainder. We believe this has the advantage of being able to reuse empirical risk and other related metrics, as in Section~\ref{ss:generalization}, but for robustness. 
Future work will thus explore expressing $\rpoprob(f)$ appropriately (\eg in terms of the relative frequency of abnormal inputs that lead to failures), as well as the relationship to the prevailing robustness related metrics. 

\subsubsection{Implications for Verification}\label{sss:verification}

As previously indicated (Section~\ref{ss:threats-to-validity}), a consequence of using the method described in Section~\ref{s:surrogate-model} is that the underlying abstraction cannot also be used to verify that the implementation meets the defined requirements. In particular, the assumptions of the binomial model cannot be used to define verification requirements because the IID constraint will not be met. 

A second implication relates to the dataset sample size $\eta$ necessary to meet the bound on the generalization gap and the related tolerance (Section~\ref{ss:generalization}). Specifically, Eq.~\eqref{eq:hoeffding} applies to any RV that can be bounded and does not depend on the underlying data distribution. Hence, $\eta$ is a pessimistic worst-case lower bound, which increases quadratically with a smaller tolerance $\epsilon$ in the generalization gap.  
Thus, alternative methods (such as using the Normal distribution approximation to the binomial) could be used to give more favorable sample sizes, subject to the validity of the assumptions of those alternative methods.

\section{Concluding Remarks}\label{s:conclusions}

\subsection{Related Work}\label{s:related-work} 

Our adaptation of the AEBS (Fig.~\ref{f:aebs-bd}) is from~\cite{aebs-dasc2024}, which is the closest counterpart to this paper. 
In Section~\ref{s:example} and~\ref{ss:new-pssa}, we clarified the modifications our paper makes to the AEBS and its safety assessment, correcting what we believe is an erroneous application of the FTA in~\cite{aebs-dasc2024}. Additionally, in \cite{aebs-dasc2024}, the Poisson distribution is used, with limited justification, to compute the per image miss probability as $\missprob \leq 0.19$, with the confirmation threshold selected as $\minhit = 5$. However, the justification for those choices is weak, although it is acknowledged that a binomial distribution is more precise. 
In contrast, we give a detailed description of the binomial model for the intended MLC and MLM behaviors, with a range of admissible values for the corresponding parameters (Section~\ref{s:surrogate-model}), also supplying substantiating rationale (Section~\ref{s:discussion}). Additionally, our paper relates the QSO to both MLM generalization and sample size estimates for the test data, thus going further than~\cite{aebs-dasc2024}.

The automotive systems domain has progressively considered the relation between system safety and MLM performance. 
For instance, in the context of detecting pedestrians, reasoning about how evidence of MLC performance contributes to safety assurance relies on showing that a required level of safety-related MLC performance has been attained in a defined operating environment~\cite{burton-safecomp-2019}. That, in turn, has been formalized within a safety assurance case framework as an assume-guarantee contract, which invokes quantitative performance requirements that an MLC must meet, under assumptions fulfilled by its inputs. 
In~\cite{burton2024}, the evaluation that such safety contracts are satisfied is further explored using \emph{subjective logic}, to account for the uncertainty in the assessment. That work also proposes that uncertainty-aware evaluation of ML performance requirements may provide the mechanisms to define suitable target values for the related metrics, such that they satisfy the system safety objectives, but stops short of clarifying what those mechanisms are. 

In the same application context of pedestrian detection,~\cite{gauerhof-safecomp-2020} emphasizes the elicitation and analysis of ML safety requirements, and their impact on assurance activities during ML development. 
That work, in turn, leverages a structured process to determine so-called \emph{validation targets}~\cite{gauerhof-validation-2018} necessary for assurance of safety of the intended functionality. Validation targets represent the evidence necessary to confirm, amongst other things, safety-related performance of machine learnt models, and the underlying insufficiencies. 
Both~\cite{gauerhof-safecomp-2020} and~\cite{gauerhof-validation-2018} give examples of machine learnt model performance and robustness requirements that impact system safety, along with rationale that clarifies the choice of specific metrics and their values. However, this clarification is limited as regards the procedures used to determine those metrics, their values, and how they follow from safety objectives. 

Assurance cases for safety functions that use supervised ML have recognized that it is a key goal to associate the safety-related properties, metrics and performance of MLCs to the higher level system safety requirements~\cite{burton-frontiers-2023}. That goal is supported by lower-level claims of completeness, consistency, and sufficiency of the referenced properties, metrics, and the associated performance limitations. The evidence for each of those, in turn, includes the results of safety analyses, systematic reviewing, safety verification, and---crucially---the definition of (valid) safety-related metrics and performance limitations. Again, indications on how latter is accomplished are notably absent. 

Other contemporary research has investigated the derivation of reliability requirements for MLCs. For example, safety-related visual transformations and changes for which human vision performance is unaffected have been used to determine and verify reliability requirements for MLCs used for machine vision~\cite{hu2022}. 
Drawing on the concepts of software \emph{operational profiles} and \emph{probability of failure on demand} (PFD),~\cite{Zhao2023} defines reliability and robustness metrics for DNN classifiers. 
In~\cite{Scheerer2024}, conformal prediction is leveraged to give a procedure to derive lower bounds on DNN reliability.
Although, in both~\cite{Zhao2023} and~\cite{Scheerer2024}, reliability modeling is a \emph{bottom-up}, component-level process, relying upon data, and iterative assessment of trained models. The work in this paper is rather a \emph{top-down} method.

A survey of contemporary approaches to specifying safety requirements for MLCs identifies various related considerations to be addressed~\cite{hawkins2025}. However, it does not specifically address how MLC safety requirements, or safety-related MLC performance requirements are to be related to or derived from system safety objectives.

Other aviation guidelines~\cite{arp4761, do178, do254, arp4754} do not give examples translating QSOs into item level performance requirements, and are inapplicable for items including ML. %However, 
There exist minimum operational performance standards (MOPS), and minimum aviation system performance standards (MASPS), that give function and application specific safety-related performance requirements, though they do not consider ML. Thus, our paper aims to mirror such efforts for machine learnt functionality, extending the state of the practice. 

Elements of the work in this paper have previously informed the ongoing effort of industry consensus-based standards committees~\cite{dp-faa-2024} (\eg EUROCAE WG-114 and SAE G-34), whose members include civil aviation regulators. However, their work is still in progress and as yet unpublished, hence we are unable to provide more details and clarification contrasting it with this paper. 

\subsection{Summary and Future Work}

The main contribution of this paper is an initial method, with rationale for validity, to mathematically translate QSOs allocated from a system safety assessment into the safety-related performance requirements and the associated metrics for an MLC and its underlying MLM. Using an example of an aircraft emergency braking system that uses a deep  neural network for runway sign detection, we have illustrated the method at the system, component, and model levels, showing the relationship to machine learnt model generalization, and sample size for test data. To the best of our knowledge, this paper is the first to systematically relate system safety and safety-related machine learnt model performance requirements.

There are several avenues to further improve upon this initial work: first, we intend to extend our method to address robustness performance (Section~\ref{ss:additional-considerations}), for example, by relaxing the assumptions on pre-processing (Section~\ref{ss:assumptions}). 
Next, we also aim to address per image false positives and false classifications, without which the current requirements are more optimistic than they should be. A candidate approach here, is to use a multinomial model, whilst also exploring Bayesian approaches, \eg with beta and Dirichlet priors, to capture and specify the uncertainty in the performance metrics. 
Additionally, the following are key to a broader applicability of our method: 
\begin{inparaenum}[(i)]
	\item addressing applications involving regression and multi-class classification problems, also considering other types of loss functions, generalization bounds, and metrics; and 
	\item addressing alternative procedures for detection confirmation, \eg using longer and/or multiple detection sequences.
\end{inparaenum}

\subsection*{Acknowledgment}

We thank various members of EUROCAE WG-114 and SAE G-34 for numerous fruitful discussions on the topic of this paper.

%
% ---- Bibliography ----
%

\bibliographystyle{IEEEtran}
%\bibliography{mlc-perf-reqs}

% Generated by IEEEtran.bst, version: 1.13 (2008/09/30)

\end{document}